\documentclass[a4paper,reqno]{amsart}
\usepackage[english]{babel}
\usepackage[T1]{fontenc}
\usepackage{amsmath,amssymb,amsfonts,amsthm}
\usepackage{hyperref}
\usepackage{slashed}
\usepackage{graphicx,color}
\usepackage[bbgreekl]{mathbbol}
\usepackage{mathtools}

\newtheorem{definition}{Definition}
\newtheorem{theorem}[definition]{Theorem}
\newtheorem{proposition}[definition]{Proposition}
\newtheorem{lemma}[definition]{Lemma}
\newtheorem{corollary}{Corollary}

\newtheorem{remark}[definition]{Remark}

\newcommand{\X}[1]{(X_{#1})}

\newcommand{\qsp}[2]{\,\ensuremath{\raise.5ex\hbox{$#1$}\big\slash\raise-.5ex\hbox{$#2$}}} 

\newcommand{\pard}[2]{\frac{\delta#1}{\delta#2}}

\newcommand{\txi}[1]{\widetilde{\xi}^{#1}}

\newcommand{\intl}{\int\limits}
\newcommand{\tc}{\widetilde{c}}
\newcommand{\tom}{\widetilde{\omega}}
\newcommand{\bom}{\boldsymbol{\omega}}
\newcommand{\te}{\widetilde{e}}
\newcommand{\bem}{\mathbf{e}}
\newcommand{\be}{\mathbf{E}}
\newcommand{\tedl}{\widetilde{\underline{e}}^\dag}
\newcommand{\bedl}{\mathbf{e}^\dag}
\newcommand{\bxi}{\boldsymbol{\xi}}
\newcommand{\bc}{\mathbf{c}}

\newcommand{\bod}{\boldsymbol{\omega}^\dag}
\newcommand{\Wedge}[1]{{\textstyle \bigwedge^{#1}}}

\title[BV-BFV GR: Palatini--Cartan--Holst Action]{BV-BFV approach to General Relativity: Palatini--Cartan--Holst action}
\author{A. S. Cattaneo}
\address{Institut f\"ur Mathematik, Winterthurerstrasse 190, 8057 Z\"urich, Switzerland}
\email{cattaneo@math.uzh.ch}

\author{M. Schiavina}
\address{Department of Mathematics, University of California, Berkeley, 970 Evans Hall, Berkeley, 94720 Berkeley, U.S.A.}
\curraddr{Institute for Theoretical Physics, ETH Z\"urich, Wolfgang Pauli Strasse 27, 8093 Z\"urich, Switzerland.}
\curraddr{Department of Mathematics, ETH Z\"urich, R\"amistrasse 121, 8092 Z\"urich, Switzerland.}
\email{micschia@phys.ethz.ch}
\date{}

\thanks{This research was (partly) supported by the NCCR SwissMAP, funded by the Swiss National Science Foundation and by the COST Action MP1405 QSPACE, supported by COST (European Cooperation in Science and Technology).
A. S. C. acknowledges partial support of SNF grant No. 200020\_172498$\slash$1. M. S. is supported by SNF grant No. P2ZHP2\_164999.
}

\begin{document}

\maketitle

\begin{abstract}
We show that the Palatini--Cartan--Holst formulation of General Relativity in tetrad variables must be complemented with additional requirements on the fields when boundaries are taken into account for the associated BV theory to induce a compatible BFV theory on the boundary. 
\end{abstract}

\tableofcontents

\section*{Introduction}
This paper deals with the BV-BFV approach to General Relativity (GR), initiated in \cite{CS1} and \cite{THESIS}. It is devoted to the Palatini--Cartan--Holst formulation of GR, following the classical analysis of \cite{CS4}. 

Elaborating on the ideas of Batalin, Fradkin and Vilkovisky (B(F)V) \cite{BV81,BFV1,BFV2}, Cattaneo, Mn\"ev and Reshetikhin (CMR) \cite{CMR1,CMR2,CMR3} suggested that in order to make sense of perturbative quantisation for a gauge theory on a manifold with boundary, suitable compatibility conditions should hold between a Lagrangian theory in the bulk manifold $M$ and its relative Hamiltonian description on the boundary $\partial M$.

At the quantum level, the mentioned compatibility is required to ensure that the state associated with the quantisation of the bulk be physical (i.e. gauge invariant), and this is formalised by requiring that it be a cocycle in a suitable complex induced by the boundary structure. At the semi-classical level, instead, one requires that the bulk action fail to be the Hamiltonian function of the BV-operator, by a term controlled by the boundary Noether one-form. One then goes on by assuming that the differential of such a boundary form be pre-symplectic (i.e. its kernel must be a subbundle), and that symplectic reduction is smooth.

When this holds, some cohomological data (a BFV-manifold) is induced on the boundary from the BV-data in the bulk, and both the algebra of constraints in the geometrical sense of Kijowski and Tulczyjew \cite{KT} and the residual gauge symmetry on the boundary are recovered from the induced \emph{boundary action} (the BFV operator). More precisely, what happens is that the Chevalley--Eilenberg--Koszul--Tate resolution of (on-shell) gauge-equivalent classes of fields, represented by the BV-data in the bulk \cite{BFV1,BFV2} (see also \cite{Stash}), gets surjectively mapped to the resolution of the (reduced) coisotropic submanifold of canonical constraints on the boundary (the BFV data \cite{Schaetz09,Schaetz10,stash97}). This produces a cohomological resolution of the canonical constraints, with the crucial property of being compatible with the original theory in the bulk. In most cases, this requires no further input than the BV-data one assigns to the bulk.

To this aim, the symplectic analysis of \cite{KT} has the great advantage, with respect to the widely employed Dirac analysis of constraints \cite{Dirac, Henn}, of being clean and of yielding a direct access to geometric or canonical quantisation and to the BV-BFV construction. 

This approach also goes in the direction of the axiomatisation of quantum field theory \cite{Ati,Seg}: the BV-BFV formalism admits a natural cutting-gluing procedure that might allow one to understand the quantum theory on elementary building blocks, to be then glued together to obtain the quantisation of a more complex space--time manifold.

After having analysed the Einstein-Hilbert formulation of General Relativity \cite{CS1}, and having shown that it does indeed satisfy the (classical) BV-BFV axioms, we now turn to another \emph{classically equivalent} formulation of GR. 

In Palatini--Cartan--Holst theory (PCH) the basic fields are a tetrad (a co-frame field) and a connection in an $SO(3,1)$ bundle\footnote{In the Euclidean case one uses $SO(4)$.}. The equivalence between EH and PCH theories is well established for closed manifolds, and its extension to the boundary is discussed in \cite{CS4}.

In Section \ref{Sect:BVPCH} we implement diffeomorphisms as gauge symmetries in the BV setting, for all theories of differential forms valued in $\mathfrak{g}$-modules, and we use it to define the BV-structure associated to PCH theory.

We find that this (natural) extension of PCH theory to the BV setting does not satisfy the BV-BFV axioms, as it does not induce a bulk-compatible BFV structure on the boundary, unless non-trivial strong requirements are imposed on the fields. We stress that this is a remarkable deviation from the Einstein--Hilbert case.

Recently, it has been shown \cite{CSS, CanSch} that the analogous theory in $(2+1)$-dimensions does not encounter such an obstruction, when inducing BFV data associated to the boundary, which then appears to be a phenomenon peculiar to dimension 4 (and possibly higher).

We plan to investigate possible solutions to this issue, e.g., by correcting the BV-form by a boundary term, a strategy that turned out successful in a similar, yet much simpler situation, when dealing with one dimensional gravity coupled to matter \cite{CS3}.

However, we argue that the usual notion of classical equivalence of field theories is insufficient to grasp differences that might arise where higher codimension data (e.g. boundaries) are taken into account. In the mentioned case of a one dimensional gravity model, a theory that is classically equivalent to the Jacobi formulation of classical mechanics is shown to enjoy a much better boundary structure than the latter, which induces a BV-BFV structure only with a careful choice of a boundary term for the BV-form \cite{CS3}.

Another possibility is presented in Section \ref{Sect:pres}, where we will replace the natural assignment of symmetries in favour of vector fields which preserve the boundary submanifold. This choice (tantamount to requiring that the vector fields have zero transverse component on the boundary) will turn out to have a great impact in the regularity of the theory. Theorems \ref{Theo:PCHpres1} and \ref{Theo:PCHpres2} will state the existence of a BV-BFV correspondence when the new BV input is considered. 

This strategy is also related to what happens in the one-dimensional examples of \cite{CS3}. In this case, though, the BFV data we obtain by considering boundary-preserving diffeomorphism will be the resolution of a coisotropic submanifold, larger than the one that defines GR.

This result poses an important question about what variational principles that describe the same Euler--Lagrange equations should be considered truly equivalent in the presence of boundary. The BV-BFV axioms might then be used as a criterion to determine whether a given variational principle has better chances than others to yield a sensible quantisation theory, if we believe that whatever quantisation eventually turns out to be, it should essentially be represented by a functorial association of a suitable category of linear objects, to the category of space--time cobordisms with structure. 

In other words, the naturality of the requirement of a bulk theory to be compatible with its boundary data, makes it hard to think that a correct notion of quantisation can be developed without taking this requirement into account.

\section{Classical BV and BFV formalisms}\label{section:background}

In this section we recall the general formalism for gravity theories, as in section 2 of \cite{CS1}.

Consider a \emph{space of fields}, i.e., a (possibly infinite dimensional) $\mathbb{Z}$-graded symplectic manifold $\mathcal{F}$ with a symplectic form $\Omega$ of degree $|\Omega|=k$ together with a local, degree $k+1$ functional $S$ of the fields and a finite number of their derivatives. 

The dynamical content of the theory is encoded in the Euler--Lagrange variational problem for the functional $S$. The $\mathbb{Z}$-grading is called \emph{ghost number}, but it will be often replaced by the computationally friendly \emph{total degree}, which takes into account the sum of different gradings when the fields belong to some graded vector space themselves (e.g. differential forms).

The symmetries are encoded by an odd vector field $Q\in \Gamma(T[1]F)$ such that $[Q,Q]=0$. A vector field with such a property is said to be  \emph{cohomological}.

Among these pieces of data some compatibility conditions are required. We give the following definitions for different values of $k$. According to the convention that we adopt, ordinary symplectic manifolds are called $(0)$-symplectic in the graded setting. Our model for a bulk theory will be given by 
\begin{definition}\label{BVdef}
A BV-manifold is the collection of data $(\mathcal{F}, S, Q,  \Omega)$ with $(\mathcal{F},\Omega)$ a $\mathbb{Z}$-graded $(-1)$-symplectic manifold, and $S$ and $Q$ respectively a degree $0$ function and a degree $1$ vector field on $\mathcal{F}$ such that
\begin{enumerate}
\item $\iota_{Q}\Omega=\delta S$, i.e. $S$ is the Hamiltonian function of $Q$
\item $[Q,Q]=0$, i.e. $Q$ is cohomological.
\end{enumerate}
\end{definition}

\begin{remark}
The symplectic structure $\Omega$ defines an odd-Poisson bracket $(,)$ on $\mathcal{F}$ and the above conditions together imply
\begin{equation}
(S,S)=0
\end{equation}
the \emph{Classical Master Equation} (CME).
\end{remark}

\begin{definition}
Whenever the data $(\mathcal{F},S,Q,\Omega)$ satisfies only (2) but not (1) we say that the BV-manifold is \emph{broken}\footnote{Sometimes one requires that $\Omega$ be only closed, allowing it to be degenerate. In this case one speaks of $(\mathcal{F}, \Omega, S, Q)$ as a \emph{relaxed} BV-manifold.}. 
\end{definition}

On the other hand, the model for a boundary theory, \emph{induced} in some sense to be explained, will be given by
\begin{definition}\label{BFVdef}
A BFV-manifold is the collection of data $(\mathcal{F}^\partial, S^\partial, Q^\partial,  \omega^\partial)$ with $(\mathcal{F}^\partial,\omega^\partial)$ a $\mathbb{Z}$-graded $0$-symplectic manifold, and $S^\partial$ and $Q^\partial$ respectively a degree $1$ function and a degree $1$ vector field on $\mathcal{F}^\partial$ such that
\begin{enumerate}
\item $\iota_{Q^\partial}\omega^\partial=\delta S^\partial$, i.e. $S^\partial$ is the Hamiltonian function of $Q^\partial$
\item $[Q^\partial,Q^\partial]=0$, i.e. $Q^\partial$ is cohomological.
\end{enumerate}
This implies that $S^\partial$ satisfies the CME. If $\omega^\partial$ is exact, we will say that the BFV-manifold is exact.
\end{definition}

These definitions abstract from the following prototype. Usually one starts from a classical theory, that is, for each manifold M of a fixed dimension, the assignment of a local action functional $S_{M}^0$ on some space of classical fields $F_M$ and a distribution in the bulk $D\subset TF_M$ encoding the symmetries, i.e. $L_X(S_{M}^0)=0$ for all $X\in \Gamma(D)$. The only requirement on $D$ for the formalism to make sense is that $D$ be involutive on the critical locus of $S_{M}^0$. Notice that $D$ can be the distribution induced by a Lie algebra (group) action, in which case it is involutive on the whole space of fields. When this is the case we will talk of BRST formalism, even though the setting will be slightly different from the original one (for another account on the relationship between the BV and BRST formalism see, e.g. \cite{PavelTh}). 

To construct a BV-manifold on a closed manifold $M$ starting from classical data we must first extend the space of fields to accommodate the symmetries: $F_M\leadsto \mathcal{F}_M=T^*[-1]D[1]$. Symmetries are considered with a degree shift of $+1$, whereas the dualisation introduces a different class of fields (called anti-fields) with opposite parity to their conjugate fields, owing to the $-1$ shift in the cotangent functor. This yields a $(-1)$-symplectic manifold, which is a good candidate to be the space of fields we want to work with\footnote{Here we assume for simplicity that $D$ can also be described in terms of local data. In more general situations, one may have to resolve $D$ into a complex described in terms of local data (ghost for ghosts).}.

The classical action has to be extended as well to a new local functional on $\mathcal{F}_M$, and if we want this to satisfy the axioms of the BV-manifold we must impose the CME on the extended action. This process of extension goes through co-homological perturbation theory \cite{BV81, Stash, stash97,FK, CMR2} and it will ensure us to end up with a BV-manifold. However, for a theory which is BRST-like, the extension is determined by the following straightforward result \cite{BV81}:
\begin{theorem}\label{minimalBV}
If $D$ comes from a Lie algebra action, the functional $S_{BV}=S_{M}^0 + \langle \Phi^\dag,Q\Phi\rangle$ on the space of fields $\mathcal{F}_M=T^*[-1]D[1]$ satisfies the CME, where $\Phi$ is a multiplet of fields in $D[1]$, $\Phi^\dag$ denotes the corresponding multiplet of conjugate (anti-)fields and $Q$ is the degree $1$ vector field encoding the symmetries of $D$. 

$\mathcal{F}_M$ is then a $(-1)$-symplectic manifold and, together with $S_{BV}$ and $Q_{BV}$ satisfying $\iota_{Q_{BV}}\Omega = \delta S$, with $\Omega$ the standard symplectic structure on cotangent bundles, it yields a BV-manifold corresponding to a (minimal) extension of the classical theory.
\end{theorem}

\subsection{BV-BFV formalism for gauge theories}\label{b(f)vgauge}

We will explain here in which sense Definition \ref{BFVdef} is a boundary model for Definition \ref{BVdef}. 

\begin{definition}\label{BV-BFV}
An exact BV-BFV pair is the quintuple $(\mathcal{F}, \Omega, S, Q; \pi)$ given by a broken BV-manifold, together with the exact BFV-manifold $(\mathcal{F}^\partial,\omega^\partial=\delta \alpha^\partial, S^\partial, Q^\partial)$ and a surjective submersion $\pi\colon \mathcal{F}\longrightarrow \mathcal{F}^\partial$ such that the BV-BFV formula:
\begin{equation}
\iota_Q\Omega=\delta S + {\pi}^*{\alpha}^\partial
\end{equation}
is satisfied. Such a pair will be denoted by $(\mathcal{F},\mathcal{F}^\partial)_{\pi}$
\end{definition}

In field theory there are natural examples of BV-BFV pairs as in the following prototypical construction. Say that we start from the data defining a BV-manifold, but this time we allow $M$ to have a boundary $\partial M$: the requirement that $\iota_{Q}\Omega=\delta S$ is (in general) no longer true. What will happen is that the integration by parts one usually has to take into account when computing $\delta S$ will leave some non zero terms \emph{on the boundary}. More precisely, consider the map 
\begin{equation}
\widetilde{\pi}_M:\mathcal{F}_M \longrightarrow \widetilde{\mathcal{F}}_{\partial M}
\end{equation}
that takes all fields and their jets to their restrictions to the boundary  (it is a surjective submersion). We can interpret the boundary terms as the pullback of a one form $\widetilde{\alpha}$ on $\widetilde{\mathcal{F}}$, namely
\begin{equation}\label{noether}
\iota_{Q}\Omega=\delta S + \widetilde{\pi}_M^*\widetilde{\alpha}
\end{equation}
We will call $\widetilde{\alpha}$ the \emph{pre-boundary} one form. In full generality $\widetilde{\alpha}$ is a connection on a line bundle, yet when $S$ is a function on the space of fields, $\widetilde{\alpha}$ is a globally well defined $1$-form.

Notice that if we are given this data, we can interpret this as a \emph{broken BV-manifold}, with some relation to the boundary. We can in fact consider the pre-boundary two form $\widetilde{\omega}\coloneqq \delta\widetilde{\alpha}$ and if it is pre-symplectic (i.e. its kernel is a subbundle) then we can define the true space of boundary fields $\mathcal{F}^\partial_{\partial M}$ to be the symplectic reduction of the space of pre-boundary fields, namely:
\begin{equation}\label{smoothred}
\mathcal{F}_{\partial M}^\partial=\qsp{\widetilde{\mathcal{F}}_{\partial M}}{\mathrm{ker}(\widetilde{\omega})}
\end{equation}
with projection to the quotient denoted by $\varpi:\widetilde{\mathcal{F}}_{\partial M}\longrightarrow \mathcal{F}_{\partial M}^\partial$. If all of the above assumptions are satisfied and the quotient $\mathcal{F}^\partial_{\partial M}$ is smooth, the map $\pi_M\coloneqq \varpi\circ\widetilde{\pi}$ is a surjective submersion, the reduced two form $\omega^\partial\coloneqq \underline{\widetilde{\omega}}$ is a $0$-symplectic form, and the key result is

\begin{proposition}[\cite{CMRCorfu}]\label{CMRprop}
The cohomological vector field $Q$ projects to a cohomological vector field $Q^\partial$ on the space of boundary fields $\mathcal{F}_{\partial M}^\partial$. Moreover $Q^\partial$ is Hamiltonian for a function $S^\partial$, the boundary action.
\end{proposition}

When this construction goes through, it associates to a manifold with boundary $(M,\partial M)$ a BV-BFV pair that depends on the manifold data. We will say that

\begin{definition}
A d-dimensional BV-BFV theory is an association of a BV-BFV pair $(\mathcal{F}_M,\mathcal{F}_{\partial M}^\partial)_{\pi_M}$ to a d-dimensional manifold with boundary $(M,\partial M)$.
\end{definition}

To summarise the construction and rephrase Proposition \ref{CMRprop} we have the following

\begin{theorem}[\cite{CMR1}]
Whenever the space $\mathcal{F}^\partial_{\partial M}$ of Equation \eqref{smoothred} is smooth, we are given the BV-BFV pair $(\mathcal{F}_M,\mathcal{F}_{\partial M}^\partial)_{\pi_M}$. The construction of a BV-manifold for a local field theory on a closed manifold $M$ extends to a (possibly exact) BV-BFV theory on the manifold with boundary $(M,\partial M)$.
\end{theorem}

\begin{remark}
Notice that in Definition \ref{BFVdef} it is possible to relax the requirement that the BFV 2-form $\widetilde{\omega}$ be nondegenerate and introduce the notion of \emph{pre-BFV} manifolds. When that is the case we may define \emph{pre-BV-BFV} pairs to be modeled over these more general pre-BFV-manifolds. Observe that a pre-BV-BFV pair $(\mathcal{F},\widetilde{\mathcal{F}})_{\widetilde\pi}$ such that the form $\mathrm{ker}(\widetilde{\omega})$ is a subbundle gives naturally rise to a BV-BFV pair on the symplectic reduction $\mathcal{F}^\partial = \widetilde{\mathcal{F}}\big\slash\mathrm{ker}(\widetilde{\omega})$, if smooth, by composing $\widetilde{\pi}$ with the symplectic reduction map.
\end{remark}

Some compatibility between bulk and boundary can always be achieved in terms of the space of pre-boundary fields, on which the differential of the Noether $1$-form is degenerate. The crucial assumption is that the symplectic reduction of this $2$-form should be smooth.

The advantage of such a point of view is at least twofold. First of all, as we just saw, the formalism is \emph{large enough} to be able to describe consistently what happens both in the bulk and in the boundary. On the other hand it is flexible enough to allow for symmetries that are more general than a Lie group action. For instance it is possible to accomodate symmetries that close only on shell (e.g. Poisson sigma model) or symmetries whose local generators are not linearly independent, where higher relations among the relations are required (e.g. BF theory or other theories involving $(d>1)$-differential forms.

Also notice that, even if we started with a Lie group action on the space of bulk fields, the BFV symmetries on the boundary may be of a more general type. This is actually the case in GR.

The BV-manifold that we have constructed in Theorem \ref{minimalBV} when a gauge theory of the BRST-kind was given is sometimes called the \emph{minimal BV-extension} of the gauge theory. When a non trivial boundary is allowed, we will use this minimal extension as the starting point for the BV-BFV analysis.

Note that the pre BV-BFV structure constructed above is not invariant under the extension to manifold with boundary of local BV diffeomorphisms. In addition, one may change the structure by changing the BV-form by a boundary contribution (which sometimes may be absorbed by a BV symplectomorphism). The different pre boundary 2-forms may have different kernels \cite{CS3}. In this paper we focus on the \emph{natural} description of PCH theory and show that the kernel does not have \emph{constant rank}; the possibility that modifications as above might solve the problem remains open.

In what follows we will check the BV-BFV axioms for the Palatini--Cartan--Holst theory of gravity and we will see that when diffeomorphisms are considered, the condition on $\widetilde{\omega}$ does indeed become an obstruction. In \cite{CS1} we proved that this step works in the Einstein--Hibert formulation of GR, in the ADM decomposition near the boundary.

Throughout the paper we will assume that $M$ is an oriented manifold that admits a Lorentzian structure.

\section{General Relativity in the Palatini--Cartan--Holst formalism}\label{Sect:tetrad}
It is possible to cast General Relativity as a theory of connections on a principal bundle, independent of the metric field. The definitions that follow are based on \cite{CS4}.

Let $\mathcal{V} \longrightarrow M$ be the Minkowski bundle over a 4-dimensional manifold $M$, with fibre the Minkowski space $(V,\eta)$, and let $P\longrightarrow M$ the associated principal $SO(3,1)$ bundle.

\begin{definition}\label{Def:PCH}
Let $F_\omega$ be the curvature of a connection $\omega\in\mathcal{A}_P$, regarded as a $\Wedge{2}V$-valued two-form under the identification $\mathfrak{so}(3,1)\simeq \Wedge{2}V$, and let $e\colon TM\longrightarrow \mathcal{V}$ be a bundle isomorphism covering the identity, i.e. $e\in\Omega_{nd}^1(M, \mathcal{V})$, where \emph{nd} stands for nondegenerate.

The \emph{Palatini--Cartan--Holst theory} is the assignment of the pair $(\mathcal{F}_{PCH}^0, S_{PCH}^0)_M$ to every $4$ dimensional manifold $M$ such that
\begin{equation}
\mathcal{F}_{PCH}^0=\Omega_{nd}^1(M, \mathcal{V}) \times \mathcal{A}_P;
\end{equation}
\begin{equation}\label{Action:PCH}
S_{PCH}^0=\intl_M \hat{T}_\gamma\left[ \frac12 e\wedge e\wedge F_\omega + \frac{\Lambda}{4} e^4 \right],
\end{equation}
where the map
\begin{equation}
\hat{T}_\gamma\colon \Wedge{2}V\otimes \Wedge{2}V {\longrightarrow} \mathbb{R}
\end{equation}
acts  as $\hat{T}(\alpha\otimes\beta) = \mathrm{Tr}\left[(1 + \star) \alpha \wedge \beta\right]$, with $\star$ the hodge operator defined by $\eta$, and the constants $\gamma,  \Lambda\in\mathbb{R}$ are respectively called Barbero--Immirzi parameter and cosmological constant.
\end{definition}

The PCH theory depends on the Barbero--Immirzi parameter $\gamma$ and the cosmological constant $\Lambda$. The role of the parameter $\gamma$ has been debated at length. We shall retain it for the sake of generality, as it will have no tangible effect in what follows, although it generates ambiguities in quantisation (see \cite{Rovth} and references therein).

We have constructed a field theory whose basic fields are a tetrad $e$ and an independent $\mathfrak{so}(3,1)$ connection $\omega$. To recover the standard Einstein--Hilbert metric formulation of GR one pulls back $\eta$ to $g\coloneqq e^*\eta$ and $\omega$ to $A\coloneqq e^*\omega$, and using the field equations imposes that $\omega=\omega(e)$ is the (unique) connection satisfying $d_\omega e=0$, which is equivalent to $\nabla_A g = 0$ and $A$ is the Levi--Civita connection.

\begin{remark}
Notice that tetrads have more local degrees of freedom than metrics (16 local against 10), but we have the gauge freedom to rotate a tetrad with a local Lorentz transformation without changing the action and the equations of motion. In addition to the usual space--time diffeomorphisms, we will take into account an internal $SO(3,1)$-symmetry as well.
\end{remark}

The Einstein--Hilbert and Palatini--Cartan--Holst theories are then equivalent \emph{on-shell}, that is they describe the same Euler--Lagrange locus, modulo symmetries. This is the content of the classical equivalence of field theories with vanishing boundary conditions.

\subsection{Classical boundary structure}
Let us summarise some of the results of the classical analysis of the boundary structure of PCH gravity presented in \cite{CS4}.

The adjective \emph{classical}, here and anywhere else in this paper, just means \emph{degree-0}. The basic procedure is equivalent to the one outlined in section \ref{b(f)vgauge}, without BV-extension. 

Denote $\mathcal{V}^\partial:=\iota^*\mathcal{V}$ the induced vector bundle on the boundary, $\iota\colon \partial M \longrightarrow M$, we also denote by $\Omega^1_{nd}(\partial M, \mathcal{V}^\partial)$ the space of $\mathcal{V}^\partial$-valued 1-forms that span a 3-dimensional subspace $W\subset V$ and by $P^\partial\equiv \iota^*P$ the induced principal bundle on the boundary. The space of restrictions of fields to the boundary projects to the (reduced) space of boundary fields, obtained as the quotient by the kernel of the map
$$\mathsf{W}_e^{(1,2)}\colon \Omega^1(\partial M, \Wedge{2}\mathcal{V}^\partial)\longrightarrow \Omega^2(\partial M, \Wedge{3}\mathcal{V}^\partial)$$
$$v\mapsto e\wedge v$$
which acts on the affine space $\mathcal{A}_{P^\partial}$, the natural space of restrictions of connections to the boundary and is surjective. More precisely, we have:
\begin{lemma}[\cite{CS4}]
\label{Lemma:kerWe}
The map 
$$\mathsf{W}_e^{(p,k)}\colon \Omega^p\left(\partial M,\bigwedge^{k}\mathcal{V}^\partial\right)\longrightarrow  \Omega^{p+1}\left(\partial M,\bigwedge^{k+1}\mathcal{V}^\partial\right)$$
defined by $\mathsf{W}^{(p,k)}_e(X)=X\wedge e$, where $e$ is the restiction of the tetrad to the boundary $\iota\colon \partial M \rightarrow M$, is injective for $p=k=1$ and it is surjective when $(p,k)=(1,2)$ or $(p,k)=(2,1)$. 
\end{lemma}

This is used to prove
\begin{theorem}[\cite{CS4}]
The classical space of boundary fields for the Palatini--Cartan--Holst theory is the symplectic manifold given by the fibre bundle:
\begin{equation}
\mathcal{F}_{PCH}^{0\partial}\longrightarrow \Omega_{nd}^1(\partial M, \mathcal{V}^\partial)\simeq T^* \Omega_{nd}^1(\partial M,\mathcal{V}^\partial)
\end{equation} 
with fibre over $e\in\Omega_{nd}^1(\partial M, \mathcal{V}^\partial)$ given by the reduction $\mathcal{A}_{P^\partial}^{red}\coloneqq\qsp{\mathcal{A}_{P^\partial}}{\sim}$ with respect to the equivalence relation $\omega\sim \omega' \iff \omega-\omega'\in \mathrm{Ker}(\mathsf{W}^{(1,2)}_e)$ and the symplectic form reads
\begin{equation}\label{classicalboundaryform}{
\varpi^{0\partial}_{PCH}=\int\limits_{\partial M}\hat{\mathrm{T}}_\gamma\left[\mathbf{e}\wedge\delta\mathbf{e}\wedge\delta\bom \right]} = \intl_{\partial M}\mathrm{Tr}\left[\delta\bem\delta\mathbf{t}_\gamma\right].
\end{equation}
with $\bom$ denoting an equivalence class in $\mathcal{A}^{red}_{P^\partial}$, and $\mathbf{t}_\gamma\coloneqq \bem\wedge T_\gamma[\bom]$.
\end{theorem}


\begin{remark}\label{splittingremark}
The global Darboux chart requires choice of a reference connection and follows from the surjectivity of $\mathsf{W}_e^{(1,2)}$. If we denote by $\mathcal{W}\coloneqq \mathrm{ker}(\mathsf{W}_e^{(1,2)})$ we can choose a complement in 
$$\Omega^1(\partial M,\Wedge{2}\mathcal{V}^\partial)=\mathcal{W} \oplus \mathcal{C}$$
and split $\omega = \widetilde{\omega} + v$, with $v\in\mathcal{W}$. Fields in $\mathcal{A}_{P^\partial}$ are then equivalence classes of $P^\partial$-connections, modulo $\mathcal{W}$.
\end{remark}

The algebra of constraints has the following structure

\begin{theorem}[\cite{CS4}]\label{constraintstheorem}
In the symplectic manifold 
$$\mathcal{F}_{PCH}^{\partial}\longrightarrow \Omega^1_{nd}(\partial M,\mathcal{V}^\partial)$$ 
with symplectic form $\varpi^\partial_{PCH}$ as in Eq. \eqref{classicalboundaryform}, the vanishing locus $\mathcal C_{PCH}$ of the functions:
\begin{equation}
\mathbf{L}_\alpha=\intl_{\partial M} \hat{T}_\gamma[\alpha\wedge \mathbf{e}\wedge d_{\bom}\mathbf{e}];\ \ \mathbf{J}_\mu = \intl_{\partial M}\hat{T}_\gamma\left[ \mu \wedge\bem\wedge F_{\bom}\right] + \mathrm{Tr}\left[\Lambda\mu\wedge\bem^3\right]
\end{equation}
with $\mu\in\Omega^0(\partial M,\mathcal{V}^\partial)$ and $\alpha\in\Omega^0(M,\bigwedge^2\mathcal{V}^\partial)$ is coisotropic. We have the algebraic structure:
\begin{subequations}\begin{align}
\{\mathbf{L}_\alpha, \mathbf{L}_{\alpha'}\} & =  \mathbf{L}_{[\alpha',\alpha]}\\
\{\mathbf{J}_\mu,\mathbf{J}_{\mu'}\} & = \mathbf{L}_{X^{\mu\mu'}}\\
\{\mathbf{L}_\alpha,\mathbf{J}_\mu\} & = \mathbf{J}_{[\alpha,\mu]} + \mathbf{L}_{H^{\mu\alpha}}
\end{align}\end{subequations}
where $X^{\mu\mu}$ and $H^{\mu,\alpha}$ are functions of the fields $\bem,\bom$ and depend on a choice of a complement of $\mathrm{Ker}(\mathsf{W}_e^{(1,2)})$.
\end{theorem} 

\begin{corollary}\label{corollaryBP}
The vanishing locus of the functions $\{\mathbf{L}_\alpha; \mathbf{J}_{\iota_\xi\bem}\}$, where $\xi$ is a vector field tangent to $\partial M$, defines a coisotropic submanifold $C_{BP}\supset C_{PCH}$. The subalgebra structure is given by
\begin{subequations}\begin{align}
\{\mathbf{J}_\xi, \mathbf{J}_{\xi'}\} & = \mathbf{J}_{[\xi,\xi']} + \mathbf{L}_{\iota_\xi\iota_{\xi'} F_{\bom_\gamma}}\\
\{\mathbf{J}_\xi, \mathbf{L}_\alpha\} & = -\mathbf{J}_{[\alpha,\iota_\xi\bem]}
\end{align}\end{subequations}
while the Hamiltonian vector field $\mathbb{J}_{\iota_\xi\bem}$ reads
\begin{equation}\label{hamiltonBP}
(\mathbb{J}_\xi)_{\bem}=- L^{\bom}_\xi\bem ;\ \  (\mathbb{J}_\xi)_{\bom} = -\iota_\xi F_{\bom}
\end{equation}

\end{corollary}

\begin{remark}
We will recover the resolution of the submanifold $C_{PB}$ in the BFV formalism in section \ref{Sect:pres}. It is important to observe that $C_{BP}$ does not describe the correct structure of General Relativity, for we eliminated one of the constraints, whereas $C_{PCH}$ does, as was shown in \cite{CS4}. The resulting reduced phase space has 3 local degrees of freedom, instead of 2.
\end{remark}

\section{Covariant BV theory}\label{Sect:BVPCH}
In this section we would like to promote Palatini--Cartan--Holst theory as presented in Definition \ref{Def:PCH} to the data of a BV-manifold. The PCH description of gravity is a BRST-like gauge theory, and it admits a BV extension, similarly to the Einstein--Hilbert version \cite{CS1}. Differently from the EH case, however, here we have to deal with an an internal $\mathfrak{so}(3,1)$ gauge freedom in addition to space--time diffeomorphisms. 

We define a covariant BV operator that represents the action of (infinitesimal) diffeomorphisms for all theories of G-connections and sections of G-associated bundles, and use it to define a solution of the Classical Master Equation for PCH gravity.

For the results in this section we will need the following:
\begin{lemma}\label{Lem:covariantL}
Let $P\longrightarrow M$ be a principal $G$-bundle and let $A\in\mathcal{A}_P$ be a connection on it. Let $\xi\in\mathfrak{X}[1](M)$ be a degree-$1$ vector field on $M$, and $\mathcal{V}$ an associated vector bundle with typical fibre the $\mathfrak{g}$ module $V_\mathfrak{g}$. For any differential form $\Phi\in\Omega^\bullet(M, \mathcal{V})$ define the covariant Lie derivative to be
\begin{equation}
L_\xi^A \Phi=[\iota_\xi,d_A]\Phi
\end{equation}
with $d_A$ being the covariant derivative induced by the connection $A$. We have the formula:
\begin{equation}
L_{[\xi,\xi]}^A\Phi - [L_\xi^A,L_\xi^A]\Phi + [\iota_\xi\iota_\xi F_A,\Phi]= 0
\end{equation}
\end{lemma} 
\begin{proof}
The proof is just a straightforward computation:
\begin{multline*}
L_{[\xi,\xi]}^A\Phi - [L_\xi^A,L_\xi^A]\Phi = L_{[\xi,\xi]}\Phi - [L_\xi,L_\xi]\Phi +  \iota_{[\xi,\xi]}[A,\Phi] + [A,\iota_{[\xi,\xi]}\Phi]\\
 - \iota_\xi d[\iota_\xi A,\Phi] - \iota_\xi[A,L_\xi^A\Phi] + d\iota_\xi[\iota_\xi A,\Phi]  +[A,\iota_\xi L_\xi^A\Phi]=\\
=2\iota_\xi d\iota_\xi[A,\Phi] - \iota_\xi\iota_\xi d[A,\Phi - d\iota_\xi\iota_\xi[A,\Phi]+[A,2\iota_\xi d\iota_\xi\Phi - \iota_\xi\iota_\xi d\Phi] \\
- 2\iota_\xi d\iota_\xi[A,\Phi] + \iota_\xi d[A,\iota_\xi\Phi] - \iota_\xi[A,\iota_\xi d\Phi - \iota_\xi[A,[\iota_\xi A,\Phi]] \\
+ \iota_\xi[A,d\iota_\xi\Phi] + d[\iota_\xi A,\iota_\xi\Phi] + [A,\iota_\xi\iota_\xi \Phi + [\iota_\xi A,\iota_\xi \Phi] - \iota_\xi d\iota_\xi\Phi]=0
\end{multline*}
as it can be carefully checked by expanding all terms. We used the odd version of the well known identity $ L_{[\xi,\xi]}\Phi - [L_\xi,L_\xi]\Phi =0$ ($\xi$ has degree 1), of which this Lemma is some covariant generalisation.
\end{proof}

\begin{lemma}\label{connectionindependence}
Under the same assumptions of Lemma \ref{Lem:covariantL}, we have that 
\begin{equation}
\iota_{[\xi,\xi]}^A\Phi := [L_\xi^A,\iota_\xi]\Phi = \iota_{[\xi,\xi]}\Phi
\end{equation}
i.e. such a combination does not depend on the connection $A$.
\end{lemma}

\begin{proof}
First, one shows that 
\[
\mathsf{B}:=2\iota_\xi[A,\iota_\xi \Phi] - \iota_\xi\iota_\xi[A,\Phi] - [A,\iota_\xi\iota_\xi \Phi] = 0
\]
since $\iota_\xi$ is a derivation of degree $0$ on (Lie algebra valued) differential forms. So we can write, adding $\mathsf{B}\equiv0$
$$\iota_{[\xi,\xi]}\Phi= 2\iota_\xi d\iota_\xi \Phi - \iota_\xi\iota_\xi d\Phi - d\iota_\xi\iota_\xi \Phi + \mathsf{B}=2\iota_\xi d_A\iota_\xi F_A - \iota_\xi\iota_\xi d_AF_A - d_A\iota_\xi\iota_\xi F_A = \iota_{[\xi,\xi]}^A\Phi$$
proving the statement. 
\end{proof}

Moreover, we have 
\begin{lemma}\label{iotaLiotaF}
Let $A$ be a connection on a principal bundle $P\longrightarrow M$ and $F_A$ its curvature form. Let $\xi\in\mathfrak{X}[1](M)$ be a degree-$1$ vector field. Then we have 
\begin{equation}
\iota_\xi L_\xi^A\iota_\xi F_A =0
\end{equation}
\end{lemma}

\begin{proof}
Observe that the contraction w.r.t. an odd vector field $\xi$ is an even operator, therefore
\begin{equation}\label{evencontra}
\iota_{[\xi,\xi]}\iota_\xi = \iota_\xi\iota_{[\xi,\xi]}
\end{equation}
Using Lemma \ref{connectionindependence}, which tells us that $[L_\xi^A,\iota_\xi]=[L_\xi,\iota_\xi]$, and since $(\iota_\xi)^3F_A=0$, together with the Bianchi identities $d_A F_A=0$, applying \eqref{evencontra} to $F_A$ we infer
\[
2\iota_\xi d_A\iota_\xi\iota_\xi F_A - \iota_\xi\iota_\xi d_A\iota_\xi F_A = 2 \iota_\xi\iota_\xi d_A \iota_\xi F_A - \iota_\xi d_A\iota_\xi\iota_\xi F_A
\]
leading to
\[
\iota_\xi d_A\iota_\xi\iota_\xi F_A = \iota_\xi\iota_\xi d_A \iota_\xi F_A
\]
which proves the statement.
\end{proof}

This will be used to prove the following
\begin{theorem}\label{Prop:genDiffeo}
Let $P\longrightarrow M$ be a principal $G$ bundle and let $A\in\mathcal{A}_P$ be a connection on it. Consider any degree 1 vector field $\xi$ on $M$, and any associated vector bundle $\mathcal{V}$ with typical fibre the $\mathfrak{g}$ module $V_\mathfrak{g}$. Denote by $\rho$ the representation on $V_\mathfrak{g}$. Let $c\in\Omega^0[1](M,\mathrm{ad}P)$ be a degree $1$ function with $\mathrm{ad}P$ the adjoint bundle to $P$, and define a vector field $Q$ on the graded manifold 
$${\mathcal{F}_M=\mathcal{A}_P\times\Omega^\bullet(M,\mathcal{V})\times\mathfrak{X}[1](M)\times\Omega^0[1](M,\mathrm{ad}P)}$$ 
through the assignment:
\begin{equation}
{\begin{array}{cc}
{Q}\, A=\iota_\xi F_A - d_\omega c & {Q}\, \Phi =L_\xi^\omega \Phi - \rho(c)\Phi \\\\
{Q}\, c = \frac{1}{2}\iota_\xi\iota_\xi F_A - \frac{1}{2}[c,c] & {Q}\, \xi = \frac{1}{2}[\xi,\xi]
\end{array}}
\end{equation}
Then $Q$ is cohomological, i.e. $[Q,Q]=0$.
\end{theorem}

\begin{proof}
We report the main steps of the various computations:
\begin{align*}
Q^2c =& \frac12 \iota_{[\xi,\xi]}\iota_\xi F_A - \frac12 d_A(\iota_\xi F_A - d_A c) - \frac12 [\iota_\xi\iota_\xi F_A,c]\\
=&\iota_\xi d_A\iota_\xi\iota_\xi F_A - \frac12\iota_\xi\iota_\xi d_A \iota_\xi - \iota_\xi\iota_\xi d_A\iota_\xi F_A  + \frac12\iota_\xi\iota_\xi[F_A,c] - [\iota_\xi\iota_\xi F_A,c]\\
=&-\iota_\xi L_\xi^A\iota_\xi F_A =0
\end{align*}
where we used Lemma \ref{iotaLiotaF}. Using Lemma \ref{Lem:covariantL} when applying $Q^2_{PL}$ to the field $\Phi$, we have
\begin{align*}
Q_{PL}^2\Phi=&\frac12L_{[\xi,\xi]}^A \Phi - L_\xi^AL_\xi^A\Phi + L_\xi^A[c,\Phi] +\iota_\xi[\iota_\xi F_A-d_Ac,\Phi] \\
-& [\iota_\xi F_A-d_Ac,\iota_\xi \Phi]+[c,L_\xi^A\Phi] -[c,[c,\Phi]] +\frac12[[c,c],\Phi]=\\
=&\frac12L_{[\xi,\xi]}^A \Phi - L_\xi^AL_\xi^A\Phi - \frac12[\iota_\xi\iota_\xi F_A,\Phi]=0
\end{align*}
whereas applying it to the connection $A$ we can use the Bianchi identity $d_AF_A=0$:
\begin{multline*}
Q^2A =\frac12\iota_{[\xi,\xi]} F_A  - \iota_\xi d_A \left(\iota_\xi  F_A  - d_A  c\right)+\frac12 d_A \left(\iota_\xi\iota_\xi  F_A  - [c,c]\right)-\left[\iota_\xi  F_A  - d_A  c,c\right]\\
=-\frac12\iota_\xi\iota_\xi d F_A  - \iota_\xi\left[A ,\iota_\xi  F_A \right] + \frac12\left[A ,\iota_\xi\iota_\xi  F_A \right]=-\frac12\iota_\xi\iota_\xi d_A   F_A  =0
\end{multline*}
We are left with $Q^2\xi=0$, which follows from the Jacobi identity in $(\mathfrak{X}(M),[,])$.
\end{proof}

This result tells us how to implement diffeomorphisms as gauge symmetries for different theories involving differential forms with values in some representation of the internal Lie algebra $\mathfrak{g}$. As we shall see below this is the case of the Palatini formulation of General Relativity.

\begin{remark}
Notice that Theorem \ref{Prop:genDiffeo} applies in the case of spin bundles as well, and we can take $\Phi$ to be a section of the associated spin bundle, thus extending the action of the diffeomorphisms in the tetrad formalism to spinors. This requires the replacement of $\mathsf{SO}(3,1)$ with its universal cover, $\mathsf{Spin}(3,1)$.
\end{remark}

\subsection{BV-extension of Palatini--Cartan--Holst theory}
In the literature, Piguet - Moritsch, Schweda and Sorella - and Baulieu and Bellon \cite{BB} all suggested a BRST operator for the PCH theory of gravity that can be summarised by the following:

\begin{proposition}[\cite{Piguet, Morschsor, BB}]\label{propPiguet}
The assignment
\begin{equation}\label{Qpig}{
\begin{aligned}
&\mathrm{s}\,e'=L_{\xi'} e' + [\theta', e']\\
&\mathrm{s}\,\omega'= L_{\xi'} \omega' + d_{\omega'}\theta'\\
&\mathrm{s}\,\xi'=\frac{1}{2}[\xi',\xi']\\
&\mathrm{s}\,\theta'=L_{\xi'} \theta' + \frac{1}{2}[\theta',\theta']
\end{aligned}}
\end{equation}
defines a cohomological vector field over 
$$F_{PP}\coloneqq \underbrace{\Omega_{nd}^1(M, \mathcal{V})}_{e'}\times\underbrace{\Omega^1(M,\Wedge{2}\mathcal{V})}_{\omega'}\times \underbrace{\mathfrak{X}[1](M)}_{\xi'} \times \underbrace{\Omega^0[1](M,\mathrm{ad}P)}_{\theta'}$$
with $\xi$ a vector field with ghost number $\mathrm{gh}(\xi)=1$ and $\theta$ a function with values in $\Lambda^2 V$ and ghost number $\mathrm{gh}(\theta)=1$. The Hamiltonian vector field of the BV-extension $S_{PP}$ of the Palatini--Cartan--Holst action by $s$ is a cohomological vector field in  $\mathcal{F}\coloneqq T^*[-1]F_{PP}$, thus the data  $(\mathcal{F}_{PP},\Omega_{PP},S_{PP},\check{s})$ defines a BV-manifold.
\end{proposition}

\begin{remark}
Observe that, in order to make sense of the formulas in Proposition \ref{propPiguet} we have to consider $\omega'$ as a global vector-valued one-form, instead of a connection. Moreover, the Lie derivatives are not \emph{covariant}. We show that there exists a version of this involving covariant expressions, which is \emph{close enough}, in the following sense.
\end{remark}

\begin{theorem}\label{Prop:PalatiniQ} 
The 4-tuple $\left(\mathcal{F}_{PCH}\coloneqq T^*[-1]\mathcal{F}_{min},\Omega_{PCH}^\gamma,Q, S^\gamma_{PCH}\right)$ defines a BV-manifold where $\mathcal{F}_{min}$ is defined as
\begin{equation}
\mathcal{F}_{min}\coloneqq \mathcal{F}^0_{PCH}\times\mathfrak{X}[1](M) \times \Omega^0[1](M,\mathrm{ad}P)\ni(e, \omega, \xi, c),
\end{equation}
$Q$ is the Hamiltonian vector field of $S_{PCH}^{\gamma}$, namely $\iota_Q\Omega^\gamma_{PCH} = \delta S^\gamma_{PCH}$, where
\begin{equation}\label{bvgravact}
{\begin{aligned}
S_{PCH}^\gamma&=\int\limits_M \hat{T}_\gamma\left[\frac12 e\wedge e\wedge  F_\omega +  \frac{\Lambda}{4} e^4 + \left(\iota_\xi  F_\omega  - d_\omega c \right)\omega^\dagger - \left([\iota_\xi ,d_\omega]e- [c,e]\right)e^\dagger\right]\\
&+\frac{1}{2}\int\limits_M\hat{T}_\gamma\left[ \left(\iota_\xi\iota_\xi  F_\omega  - [c,c]\right)c^\dagger\right] + \int\limits_M \frac12\iota_{[\xi,\xi]}\xi^\dagger, 
\end{aligned}}
\end{equation}
$F_\omega$ being the curvature of $\omega$, $L_\xi^\omega = [\iota_\xi ,d_\omega]$ the covariant Lie derivative along $\xi$ with connection $\omega$, and the standard $(-1)$-symplectic form $\Omega^\gamma_{PCH}$ is
\begin{equation}
\Omega_{PCH}^\gamma = \intl_M \hat{T}_\gamma\left[\delta \omega^\dag \delta\omega + \delta e^\dag \delta e + \delta c^\dag\delta c\right] + \iota_{\delta\xi}\delta\xi^\dag.
\end{equation} 

The BV operator for the PCH formalism is then given by:
\begin{equation}\label{Qpal}
{\begin{array}{cc}
{Q}\, \omega=\iota_\xi F_\omega - d_\omega c & {Q}\, e =L_\xi^\omega e - [c,e] \\\\
{Q}\, c = \frac{1}{2}\iota_\xi\iota_\xi F_\omega - \frac{1}{2}[c,c] & {Q}\, \xi = \frac{1}{2}[\xi,\xi].
\end{array}}
\end{equation}

Finally, there is a canonical transformation between the BV-manifold just described and $(\mathcal{F}_{PP},\Omega_{PP},S_{PP},\check{s})$ (cf. Proposition \ref{propPiguet}), i.e.
\begin{equation}
\phi\colon \mathcal{F}_{PCH} \longrightarrow \mathcal{F}_{PP},
\end{equation}
whose generating function is given by
\begin{equation}
G[c^\dag,\xi^\dag,e^\dag,\omega^\dag,e',\omega',\xi',\theta']\coloneqq \int_M \mathrm{Tr}\left[c^\dag(\iota_{\xi} \omega' - \theta') + \iota_{\xi}\xi^\dag{}' - e^\dag e' - \omega^\dag \omega'\right].
\end{equation}
\end{theorem}

\begin{remark}[Graded canonical transformations]\label{cantransf}
Observe that in the graded setting the generating function of a canonical trasformation might incur in some nontrivial sign conventions.

Denote by $(p,q)$ and $(P,Q)$ two Darboux charts of an odd symplectic manifold, e.g. $\Omega=\delta p \delta q$ and $|p|=|q|+1$. Now, a \emph{generating function of type I} is a function $F(q,Q)$ such that 
$$
p\delta q = P\delta Q + \delta F(q,Q) = P\delta Q + \delta q\pard{F}{q} + \delta Q\pard{F}{Q}
$$
and since $q\,\delta p = (-1)^{|q|+1}\delta q\, p$ we have the equations
\begin{equation}
p= - (-1)^{|q|}\pard{F}{q};\ \ \ \ P =  (-1)^{|Q|}\pard{F}{Q}.
\end{equation}
Alternatively, we may consider the class of \emph{generating functions of type II}, i.e. functions $G(q,P)$ that satisfy (we use the convention $\delta = \delta q \pard{}{q} + \delta P\pard{}{P}$)
$$
p\delta q = P\delta Q + \delta\left((-1)^{|P|+1}P\,Q + G \right) = (-1)^{|P|+1}\delta P\,Q + \delta q \pard{G}{q} + \delta P\pard{G}{P}
$$
and the associated equations become
\begin{equation}
p= - (-1)^{|q|}\pard{F}{q};\ \ \ \ Q =  (-1)^{|P|}\pard{F}{P}.
\end{equation}

In particular, in this class we have the generating function for the identity map, which can be easily shown to be 
\begin{equation}
G_{id} = (-1)^{|P|} P\, q.
\end{equation}

Observe that one might get rid of (some of) the signs by means of derivation from the right and by using them to define the operator $\delta$. Hereinafter we will assume that $\delta = \sum_\phi \delta\phi \pard{}{\phi}$ and we will consider its total degree to be $1$.
\end{remark}

\begin{proof}[Proof of Theorem \ref{Prop:PalatiniQ}]
$Q$ encodes the symmetries of the classical PCH action, indeed $QS^0_{PCH}=0$ :
\begin{multline*}
QS_{PCH}^0=\frac12\int\hat{T}_\gamma\left[ 2[\iota_\xi,d_\omega]e e  F_\omega  - 2[c,e]e F_\omega  - ee d_\omega (\iota_\xi  F_\omega  - d_\omega c)\right] \\
=\frac12\int\hat{T}_\gamma\Bigg[-2d_\omega e \iota_\xi e  F_\omega  - 2d_\omega e e \iota_\xi  F_\omega  - 2d_\omega \iota_\xi e e  F_\omega  - ee d_\omega\iota_\xi  F_\omega \\ 
+ ee[ F_\omega ,c] - 2[c,e] e F_\omega  + 2ed_\omega\iota_\xi e  F_\omega  + e\iota_\xi e d_\omega  F_\omega  - 2 d_\omega e e \iota_\xi  F_\omega \\ 
 - 2d_\omega \iota_\xi e e  F_\omega  + 2 d_\omega e e\iota_\xi  F_\omega  - \langle ee,\mathrm{ad}_c  F_\omega \rangle - \langle \mathrm{ad_c}(ee), F_\omega \rangle\Bigg] =0
\end{multline*}

On the other hand, $Q$ is cohomological because of Theorem \ref{Prop:genDiffeo}, where $\bigwedge^2V\simeq\mathfrak{g}=\mathfrak{so}(3,1)$, $A=\omega$ and $(V,\eta)$ clearly bears a representation of $\mathfrak{g}$. Since the symmetries are BRST-like we can use Theorem \ref{minimalBV} to construct $S_{PCH}^{BV}$, which solves the Classical Master Equation w.r.t. $\Omega_{BV}^\gamma$.

Now, let us turn to the generating function 
$$G\coloneqq \int_M \mathrm{Tr}\left[c^\dag(\iota_{\xi} \omega' - \theta') + \iota_{\xi}\xi^\dag{}' - e^\dag e' - \omega^\dag \omega'\right],
$$
from which we can deduce the following map (cf. Remark \ref{cantransf}, where we set $q=(c^\dag,\xi,e,\omega^\dag),\ P=(e^\dag{}',\omega',\xi^\dag{}',\theta')$)
\begin{equation}\begin{cases}
e'=e \\ 
\omega'=\omega \\
\xi'=\xi \\
\theta'=\iota_{\xi }\omega - c
\end{cases}
\begin{cases}
e^\dag{}' = e^\dag \\
\omega^\dag{}'  =  \omega^\dag  - \iota_{\xi}c^\dag  \\
\xi^\dag{}' = \xi^\dag + c^\dag\omega_\bullet\\
\theta^\dag{}' = - c^\dag 
\end{cases}\end{equation}
where $c^\dag\omega_\bullet$ denotes a top-form valued one-form (so contractions act as  $\iota_X(c^\dag \omega_\bullet) = c^\dag\iota_X\omega$). We can pullback $S_{PP}$ along $\phi_G$ as:
\begin{multline}
\phi_G^*S_{PP}=\phi_G^*\intl_M \hat{\mathrm{T}}_\gamma\Bigg[\frac12 e'e'F_{\omega'} +\frac{\Lambda}{4}e'{}^4 + e^\dag{}' \left( L_{\xi'}e' + [\theta',e']\right) + \\
+ \omega^\dag{}'\left( L_{\xi'}\omega' + d_{\omega '}\theta'\right)  + \theta^\dag{}'\left(\frac12[\theta',\theta'] + L_{\xi}\theta'\right)\Bigg] + \frac12\iota_{[\xi',\xi']} \xi^\dag{}'=\\
=\intl_M\hat{\mathrm{T}}_\gamma\Bigg[\frac12 eeF_\omega +\frac{\Lambda}{4}e^4 + e^\dag\left(L_\xi e + [\iota_\xi\omega,e] - [c,e]\right) + \omega^\dag\left(L_\xi \omega + d_\omega \iota_\xi\omega - d_\omega c\right) +  \\
+ c^\dag\left(\iota_\xi L_\xi\omega +\iota_\xi d_\omega\iota_\xi\omega - \frac12 [\iota_\xi\omega,\iota_\xi\omega] -\frac12[c,c] -\iota_\xi d\iota_\xi \omega + \frac12\iota_{[\xi,\xi]}\omega\right)\Bigg] + \frac12 \iota_{[\xi,\xi]}\xi^\dag =\\
=\intl_M \hat{\mathrm{T}}_\gamma\left[\frac12 eeF_\omega +\frac{\Lambda}{4}e^4 + e^\dag\left(L_\xi^\omega e - [c,e]\right) + \omega^\dag\left(\iota_\xi F_\omega - d_\omega c\right) + \frac12c^\dag\left(\iota_\xi\iota_\xi F_\omega - [c,c]\right)\right] + \frac12\iota_{[\xi,\xi]}\xi^\dag 
\end{multline}
And the last line coincides with $S_{PCH}$. We used the fact that $L_\xi \omega + d_\omega \iota_\xi\omega=\iota_\xi d\omega + [\omega,\iota_\xi\omega] = \iota_\xi F_\omega$, as well as  (observe the difference between $d_\omega$ and $d$)
$$\iota_\xi L_\xi\omega +\iota_\xi d_\omega\iota_\xi\omega - \frac12 [\iota_\xi\omega,\iota_\xi\omega] -\iota_\xi d\iota_\xi \omega + \frac12\iota_{[\xi,\xi]}\omega = \frac12\iota_\xi\iota_\xi F_\omega$$
to conclude the argument.
\end{proof}

\begin{remark}
Observe that one can map $c\mapsto -c$ and promote this to a canonical transformation by replacing, in the generating function of the identity, the sign reversal term $F=c^\dag c $.
\end{remark}

For the sake of clarity we summarise the nature of the fields, anti-fields, ghosts and anti-ghosts in $\mathcal{F}_{PCH}=T^*[-1]\mathcal{F}_{min}$ in the following table:
\begin{equation}
\begin{array}{c|c|c|c|c}
\text{Field} & \Omega^\bullet(M) & \Lambda^\bullet V & \text{Ghost} &\text{Total Degree} \\\hline
\omega & 1 & 2 & 0& 3\\
e & 1 & 1 & 0 & 2\\
c & 0 & 2 & 1 & 3\\
\xi & n.a. & n.a. & 1 & 1\\
\omega^\dag & 3 & 2 & -1 & 4\\
e^\dag & 3 & 3 & -1 & 5\\
c^\dag & 4 & 2& -2 & 4\\
\xi^\dag & 1\otimes 4 & n.a. & -2 & 3
\end{array}
\end{equation}
The ghost field $\xi$ is a vector field on $M$, and its dual anti-ghost is a one form with values in top forms, i.e. with $\chi\in\Omega^1(M)[-2]$ and $v$ a top form:
\begin{equation}
\xi^\dag=\chi v.
\end{equation} 

\section{BV-BFV approach to Palatini--Cartan--Holst theory}

We are now ready to establish whether the BV theory \eqref{bvgravact} obtained by minimally extending the Palatini--Cartan Holst action does satisfy the BV-BFV axioms or not.

\begin{theorem}\label{Theo:HolstNOGO}
The BV data $(\mathcal{F}_{PCH}, S_{PCH}, Q, \Omega_{BV}^\gamma)$ on a $(3+1)$-dimensional pseudo-Riemannian manifold $M$ with boundary $\partial M$ does not yield a BV-BFV theory for any value of $\gamma$, including the limiting case $\gamma \rightarrow \infty$, which yields the usual Palatini--Cartan formulation of gravity.
\end{theorem}

\begin{proof}
The variation of $S_{PCH}$ reads as follows:
\begin{equation}\begin{aligned}
\delta S_{PCH}^{BV}&= \int\limits_{\partial M} \hat{T}_\gamma\left[-\frac12 ee\delta  \omega  + \delta  \omega  (\iota_\xi  \omega ^\dagger) + \delta c  \omega ^\dagger + \delta e(\iota_\xi e^\dagger) + (\iota_{\delta\xi} e)e^\dagger \right]\\
&+ \int\limits_{\partial M}\hat{T}_\gamma\left[ - (\iota_\xi\delta e)e^\dagger - \frac{1}{2}\delta  \omega  (\iota_\xi\iota_\xi c^\dagger)\right]  + \int\limits_{\partial M}(\iota_{\delta\xi}\chi) \iota_\xi v + \int\limits_M\text{Bulk Terms}
\end{aligned}\end{equation}

In fact, the variation of the $\xi$-ghost part is computed as:
\begin{align}\notag
\delta \int\limits_M \frac12\iota_{[\xi,\xi]}\xi^\dag =& \intl_M\delta\left(\iota_\xi d \iota_\xi - \frac12\iota_\xi\iota_\xi d\right)\xi^\dag - \frac12 \iota_{[\xi,\xi]}\delta\xi^\dag \\\notag
=&\intl_M\iota_{\delta\xi}\left(d\iota_\xi\xi^\dag - \iota_\xi d\xi^\dag\right) - \iota_\xi d\iota_{\delta\xi}\xi^\dag - \frac12 \iota_{[\xi,\xi]}\delta\xi^\dag \\
=&-\intl_M\iota_{\delta\xi} L_\xi\left(\chi v\right) + \frac12 \iota_{[\xi,\xi]}\delta\xi^\dag + \intl_{\partial M} \iota_{\delta\xi}\chi \iota_\xi v
\end{align}

If we denote by $\xi^n$ the transversal part of $\xi$ with respect to the boundary, and with $v^\partial$ a volume form on the boundary, we may rewrite $\iota_\xi v=v^\partial\xi^n=-\xi^n v^\partial$, since $\mathrm{dim}(\partial M) =3$.

To obtain the pre-boundary one form $\widetilde{\alpha}$ we must consider the restriction of the fields to the boundary and their possible residual transversal components. With an abuse of notation, the restriction of the fields to the boundary will be denoted by the same symbol, whereas an apex ${}^n$ will be assigned to the transversal components. For instance, we will write $\iota_\xi\phi\big|_{\partial M} = \iota_{\xi^\partial}\phi^\partial + \phi_n\xi^n\equiv \iota_{\xi}\phi + \phi_n\xi^n$ by renaming the restrictions to the boundary $\phi^\partial \equiv \phi$ where $\phi$ is any field, and $\xi^\partial\equiv\xi$. We obtain
\begin{equation}\begin{aligned}\label{PCHpreone}
\widetilde{\alpha}&= \int\limits_{\partial M} \hat{T}_\gamma\left[ - \frac12 ee\delta  \omega + \delta  \omega  (\iota_\xi  \omega ^\dagger) +\delta \omega \,\omega^\dag_n \xi^n + \delta c\,  \omega ^\dagger -\delta e\, e^\dag_n \xi^n - \delta e(\iota_\xi e^\dagger)\right]\\    
&+ \int\limits_{\partial M}\hat{T}_\gamma\left[ -\delta( e_n\xi^n)e^\dag -  \delta(\iota_\xi e)e^\dagger - \delta  \omega\, (\iota_\xi c_n^\dagger) \xi^n \right] - \xi^n\iota_{\delta \xi} \chi\,v^\partial
\end{aligned}
\end{equation}
and we may compute the pre-boundary 2-form $\widetilde\varpi=\delta\widetilde\alpha$ to be $(\rho=1,2,3,n)$
\begin{equation}\label{PCHpretwo}
\begin{aligned}
\widetilde\varpi=\int\limits_{\partial M}&\hat{T}_\gamma\Bigg[- \delta e e \delta  \omega - \delta  \omega (\omega ^\dagger_\rho \delta\xi^\rho) + \delta \omega\delta\omega^\dag_\rho \xi^\rho  + \delta c\delta  \omega ^\dagger  + \delta e \delta(e^\dagger_n \xi^n) \\ 
&    + \delta( e_n \xi^n) \delta e^\dag +  \delta(e_a  e^\dag)\delta\xi^a  + \delta  \omega\, \delta\xi^n \iota_\xi c_n^\dagger -\delta\omega\,\xi^n \iota_{\delta\xi}c^\dag_n - \delta\omega\,\xi^n \iota_\xi\delta c^\dag_n \Bigg] +\\
&   +\left(\xi^n\delta\xi^n\delta\chi_n - \delta\xi^n\delta\xi^n \chi_n - \delta\xi^n\chi_a\delta\xi^a + \xi^n \delta\chi_a\delta\xi^a\right)v^\partial
\end{aligned}
\end{equation}
The kernel of $\widetilde{\varpi}$ is defined by the equations:
\begin{subequations}\label{KerPCHinit}\begin{align}
&(X_{\omega^\dag})=0\\
&(X_c)=\iota_\xi(X_\omega) \\\label{palkerXxi}
&(X_{\xi^\rho})e_\rho + (X_{e_n})\xi^n=0 
\end{align}\end{subequations}
together with
\begin{align}\label{palkerXomega}
(X_\omega)\wedge e =\Omega \\ \label{palkerXe}
(X_e)\wedge e = \mathcal{E}
\end{align}
where
\begin{equation}
\Omega\coloneqq \left[(X_{\xi^n})e^\dag_n   + (X_{e^\dag_n})\xi^n + \iota_{(X_\xi)}e^\dag\right] 
\end{equation}
\begin{equation}\label{Ecoeff}
\mathcal{E}\coloneqq \left[(X_{\omega^\dag_n})\xi^n  - (X_{\xi^\rho})\omega^\dag_\rho + (X_{\xi^n})\iota_\xi c^\dag_n - \iota_{(X_\xi)} c^\dag_n \xi^n - \iota_\xi(X_{c^\dag_n})\xi^n\right]
\end{equation}
with $\Omega\in \Omega^2(\partial M,\bigwedge{}^{\!\! 3}\mathcal{V})$ and $E\in \Omega^2(\partial M,\bigwedge{}^{\!\! 2}\mathcal{V})$. In addition we have
\begin{align}\label{edagn}
\hat{T}_\gamma\left[ (X_e)e^\dag_n - (X_\omega)\omega^\dag_n - (X_{e^\dag})e_n + (X_\omega) c^\dag_{na}\xi^a\right] + \\\notag
- \left(2(X_{\xi^n})\chi_n + (X_{\chi_n})\xi^n +(X_{\xi^a})\chi_a\right)v^\partial=0\\\label{edaga}
\hat{T}_\gamma\left[ (X_e)e^\dag_a - (X_\omega)\omega_a^\dag - (X_\omega)c^\dag_{na}\xi^n - (X_{e^\dag})e_a\right] + \\\notag
-\left( (X_{\xi^n})\chi_a  + (X_{\chi_a})\xi^n\right)v^\partial= 0
\end{align}
where the latter is valid for all $a=1,2,3$. Finally, for all $\rho=1\dots 4$
\begin{subequations}\label{nonsing}
\begin{align}\label{Xomegacheck}
(X_\omega)\xi^n =0\\\label{Xecheck}
(X_e)\xi^n=0\\ \label{Xomegacheck2}
\iota_\xi(X_\omega)\xi^n=0\\
(X_\xi)^\rho\xi^n=0 \label{Xxicheck}
\end{align}
\end{subequations}

Equation \eqref{palkerXe} is singular. As a matter of fact, counting the number of unknowns (the $(X_e)_{a}^{i}$ are 12, independent fields) against the number of equations (the $\delta \omega_{a}^{ij}$ are 18 independent variations) it is easy to gather that the system admits solutions only when relations among the $\mathcal{E}$ coefficients \eqref{Ecoeff} are imposed. On the other hand, such relations are singular for they involve polynomial expressions of odd fields only.

A more direct way to see this is by using the splitting $\omega = \tom + v$ of Remark \ref{splittingremark}, where $v\in\mathcal{W}\equiv\mathrm{ker}\mathsf{W}_e^{(1,2)}$. Then, equation \eqref{palkerXe} splits into an invertible part (the coefficient of $\delta \tom$),  and a singular part (the coefficient of $\delta v$). In fact, from the splitting 
$$\mathcal{U}\equiv\Omega^1(\partial M, \Wedge{2}\mathcal{V}) = \mathcal{W} \oplus \mathcal{C}$$
we induce the dual splitting $\mathcal{U^*}=\mathcal{W}^*\oplus \mathcal{W}^0$, with $\mathcal{W}^0$ the annihilator of $\mathcal{W}$ and the identification $\mathcal{W}^0\simeq \mathrm{im}\mathsf{W}^{(1,1)}\simeq e\wedge \Omega^1(\partial M, \mathcal{V})$. We can then project $E$ onto $\mathcal{C}^*$, and solve for $\X{\widetilde\omega}$. Observe that $\X{\tom}$ is proportional to $\xi^n$. However, the equation coming from the vanishing of the coefficient of $\delta v$ enforces a singular relation in the $\mathcal{E}$'s. 

Moreover, consider equations \eqref{Xomegacheck} and \eqref{Xomegacheck2}, and use again the splitting. Since $\X{\tom}\propto\xi^n$ the respective parts of \eqref{Xomegacheck} and \eqref{Xomegacheck2} are automatically satisfied, while $\X{v}\xi^n$ and $\iota_\xi \X{v}\xi^n$ are singular, for $\X{v}$ is a free parameter.

Thus, the kernel of $\widetilde{\varpi}$ does not define a sub-bundle of the tangent space to the space of fields, and symplectic reduction cannot be performed.
\end{proof}

\begin{remark}
This result is hinting at the fact that, for the BV-extended theory to be compatible with the boundary, we would need to require some conditions on the fields (e.g. $v=0$, that is $\omega\big|_{\partial M}\in \mathcal{C}$). 
\end{remark}

There are examples of classically equivalent theories that fail to be equivalent at the BV level when boundaries are included, see for instance \cite{CS3}. This is another nontrivial example of a similar phenomenon.

\begin{remark}
Observe that the choice of BV extension we make here is natural. As a matter of fact, looking at the three-dimensional analogue of Palatini--Cartan theory\footnote{The Holst term does not apply in dimension three, but this fact has no bearing on the argument we are making here.} one can show not only that it is strongly equivalent to BF theory \cite{CSS}, but that the same general BV extension we present in Theorem \ref{Prop:genDiffeo} and Theorem \ref{Prop:PalatiniQ}, extends to boundaries, corners and vertices \cite{CanSch}.
\end{remark}

To overcome the problem encountered with Theorem \ref{Theo:HolstNOGO}, one could try to look at a different BV input, following Corollary \ref{corollaryBP}, in the spirit of what was done for the reparametrisation invariant Jacobi theory. This is done in section \ref{Sect:pres}.

\begin{remark}
Another option might be to find a correction of the BV-form by a boundary term that changes the kernel of the pre-boundary two-form. Note that one problem are equations \eqref{Xomegacheck}, \eqref{Xomegacheck2}, which put extra, singular constraints on the kernel of $\widetilde{\varpi}$ in the omega directions, whereas from the classical analysis \cite{CS4} we expect $X_\omega$ to be the kernel of the map $e\wedge$. One might look for a boundary term for the BV-form that removes precisely these extra constraints.
\end{remark}

\section{Boundary-preserving BV-data}\label{Sect:pres}
In this section we will adopt a different point of view when looking at the symmetry distribution for the Palatini--Holst formulation of GR. Differently from what we have done in the previous section, and in the case of GR in the Einstein--Hilbert formalism \cite{CS1}, we will now consider the symmetry distribution given by the action of all the spacetime diffeomorphism of the manifold with boundary $M$ that preserve the boundary submanifold $\partial M$. At the infinitesimal level this means considering all vector fields in $M$ that are tangent at the boundary. In other words, we will assume that $\xi^n\big|_{\partial M}=0$. Let us denote by $\mathfrak{X}(M,\partial M)$ the space of such vector fields, the new space of fields we will consider is simply given by $\mathcal{F}_{PCH^\partial}\coloneqq T^*[-1]\mathcal{F}_{BP}$ where
\begin{equation}\label{PCHboundpresfields}{
\mathcal{F}_{BP}\coloneqq\underbrace{\Omega_{nd}^1(M, \mathcal{V})}_{e}\times\underbrace{\mathcal{A}_P}_{\omega}\times \underbrace{\mathfrak{X}[1](M,\partial M)}_\xi \times \underbrace{\Omega^0[1](M,\mathrm{ad}P)}_{c}}
\end{equation}

Using the same BV-extended action \eqref{bvgravact} we analysed in section \ref{Sect:BVPCH} we can consider the BV-manifold obtained by choosing the boundary-preserving space of fields \eqref{PCHboundpresfields}. We obtain the following

\begin{theorem}\label{Theo:PCHpres1}
The BV-manifold defined by the data $(\mathcal{F}_{PCH^\partial},S_{PCH}, Q,\Omega_M)$ satisfies the CMR axioms and yields an exact BV-BFV theory.
\end{theorem}

\begin{proof}
Since we required $\xi^n\big|_{\partial M}=0$ the space of pre-boundary fields will only contain the restrictions of fields to the boundary, without the residual normal directions coming from the contractions with the ghost vector field, i.e. 
\[
(\iota_\xi e)\big|_{\partial M} = (e_a\xi^a + e_n\xi^n)\big|_{\partial M}= (e_a\xi^a)\big|_{\partial M}
\]
Starting from the expression for the pre-boundary one-form found in Theorem \ref{Theo:HolstNOGO} and omitting the components of the fields along the normal direction we obtain:
\begin{equation}\label{preoneBS}
\widetilde{\alpha}=\intl_{\partial M}  \hat{T}_\gamma\left[\frac12 ee\delta\omega-\iota_\xi\delta\omega\omega^\dag + \delta c\omega^\dag -\iota_{\delta\xi}ee^\dag\right]
\end{equation}
and
\begin{equation}
\widetilde{\varpi}=\intl_{\partial M}\hat{T}_\gamma\left[\delta e e \delta \omega - \iota_{\delta\xi}\delta\omega\omega^\dag - \iota_\xi\delta\omega\delta\omega^\dag + \delta c\delta\omega^\dag +\iota_{\delta\xi}\delta e e^\dag + \iota_{\delta\xi}e \delta e^\dag\right]
\end{equation}
Collecting the various terms along the field directions we get the following equations
\begin{align}\label{Xxizero}
\delta e^\dag\colon& \X{\xi}^ae_a =0 \\\label{Xomegadagzero}
\delta c\colon& \X{\omega{^\dag}}=0
\end{align}
Solving \eqref{Xxizero} we obtain that $\X{\xi}=0$, which together with \eqref{Xomegadagzero} will simplify the remaining kernel equations to yield:
\begin{align}\label{Xomegapres}
\delta e\colon& T_\gamma[\X{\omega}]e  = 0\\\label{Xepres}
\delta \omega\colon& \X{e} e =0\\\label{Xedagpres}
\delta\xi^a \colon& \X{e^\dag}e_a = \X{e}e_a^\dag + \X{\omega}_a\omega^\dag\\\label{Xcpres}
\delta \omega^\dag\colon&\X{c}=\iota_\xi\X{\omega}
\end{align}
Equations \eqref{Xomegapres} and \eqref{Xepres} are the statement that $T_\gamma[\X{\omega}]\in\mathrm{ker}\mathsf{W}_e^{(1,2)}$ and $\X{e}\in\mathrm{ker}\mathsf{W}_e^{(1,1)}$, while the remaining equations \eqref{Xedagpres} and \eqref{Xcpres} are determined by the values of $\X{e}$ and $\X{\omega}$. 

Recalling Lemma \ref{Lemma:kerWe} and using the splitting defined in Remark \ref{splittingremark}, we have that 
\begin{equation}
\X{e}=0;\ \ \  \X{\omega}\in \mathcal{C}.
\end{equation}

It is possible to see that some of the components of $\X{e^\dag}$ are free: expand the left hand side of equation \eqref{Xedagpres} in the basis $\{e_\mu\}$ to get the simplification
\[
\X{e^\dag}^{\mu\nu\rho} e_\mu e_\nu e_\rho e_a = \X{e^\dag}^{ncb}e_ne_c e_b e_a
\]
The components $\X{e^\dag}^{abc}$, for a total of $3$ functions, do not appear anywhere in the kernel equations and are therefore free. The residual condition on the $\X{e^\dag}$ reads then $\forall c=1,2,3$
\begin{equation}\label{tempedag}
\X{e^\dag}^{nab} e_ne_ae_be_{c} = \X{\omega}^{\mu\nu}\omega_{c}^\dag{}^{\rho\sigma} e_\mu e_\nu e_\rho e_\sigma
\end{equation}

We define a new field by $\underline{e}^\dag_a = e_ae^\dag=e_ae^\dag{}^{nbc}e_ne_be_c$, which can be thought of as $\underline{e}^\dag=e\otimes e^\dag\in\Omega^1(\partial M)\otimes \Omega^{\text{top}}(\partial M,\bigwedge^4\mathcal{V})$, and the equation reads
\[
\X{\underline{e}^\dag}_a=\X{\omega}_a\omega^\dag
\]

The kernel of the pre-boundary 2-form $\varpi$ is then generated by the vector fields:
\begin{subequations}\label{vertPCHpres}\begin{align}
\mathbb{E}^\dag{}=&\X{e^\dag{}^{abc}}\pard{}{e^\dag{}^{abc}}\\
\mathbb{\Omega}=&T_\gamma[\X{v}]\pard{}{T_\gamma[v]} + \iota_\xi\X{\omega}\pard{}{c} + \X{\omega}_a\omega^\dag\pard{}{\underline{e}^\dag_a},
\end{align}\end{subequations}
where we used the splitting $T_\gamma[\X{\omega}] = T_\gamma[\X{\tom} + \X{v}]$, with $T_\gamma[\X{v}]$ in $\mathrm{Ker}(\mathsf{W}_e^{(2,1)})$, and it is then a smooth subbundle of $T\widetilde{\mathcal{F}}$. It is a matter of an easy check to show 
\[
\iota_{\mathbb{E}^\dag}\widetilde{\alpha}=\iota_{\mathbb{\Omega}}\widetilde{\alpha}=0
\]
and prove that $\widetilde{\alpha}$ is basic.
\end{proof}

We can push our understanding of the boundary structure a little bit further and obtain the explicit expressions of the BFV data.
\begin{theorem}\label{Theo:PCHpres2}
The BV-BFV pair $(\mathcal{F}_{PCH^\partial},\mathcal{F}_{PCH^\partial}^\partial)_{\pi_{PCH^\partial}}$ is such that the space of boundary fields $\mathcal{F}_{PCH^\partial}^\partial$ is the exact symplectic manifold 
\begin{equation}
\mathcal{F}^\partial_{PCH^\partial}=\mathcal{F}^{0\partial}_{PCH} \times T^*\left(\mathfrak{X}[1](\partial M)\times\Omega^0(\partial M,\iota^*\mathrm{ad}P)\right),
\end{equation}
the surjective submersion $\pi_{PCH^\partial}$ reads
\begin{equation}\label{PCHmapcovariant}
\pi_{PCH^\partial}\colon\begin{cases}
\be=- \frac12 e\wedge e + \iota_{\xi}\omega^\dag\\
\bxi{} =\xi\\
\bod=\omega^\dag\\
\bc=c - \iota_\xi\mathrm{v}\\
\bedl=e\otimes e^\dag - \mathrm{v}\otimes\omega^\dag\\
\bom=\omega - \mathrm{v}
\end{cases}
\end{equation}
where $T_\gamma[\mathrm{v}]\in\mathrm{ker}(\mathsf{W}^{(1,2)}_e)$, and $\bom\in\mathcal{A}_{\iota^*P}^{red}\coloneqq\qsp{\mathcal{A}_{\iota^*P}}{\mathrm{Ker}(\mathsf{W}^{(1,2)}_e)}$ is a connection on the boundary. 
In this chart, the boundary 2-form reads
\begin{equation}
\varpi^\partial = \intl_{\partial M}  \mathrm{Tr}\left[\delta\be\delta \bom + \delta \bc\delta\bod + \iota_{\delta\bxi{}}\delta\bedl\right]
\end{equation}
and the boundary action:
\begin{equation}\label{PCHpresBaction}
S^\partial=\intl_{\partial M} \mathrm{Tr} \left[ - \bc d_{\bom}\be +\iota_{\bxi}(F_{\bom} - \Lambda \be) \be  + \frac12\left(\iota_{\bxi}\iota_{\bxi}F_{\bom} - [\bc,\bc]\right)\bom^\dag - \frac12\iota_{[\bxi,\bxi]}\bedl \right]
\end{equation}

\end{theorem}

\begin{proof}
The proof goes through by finding an explicit expression for the map $\pi_{PCH^\partial}$  through the flows of the vertical vector fields \eqref{vertPCHpres}. We can use $\mathbb{E}^\dag$ to set $e^\dag{}^{abc}e_ae_be_c$ to zero, and $T_\gamma\left[\mathbb{\Omega}\right]$ to set $T_\gamma[v]=0$.

We get:
\[
\X{c}=\iota_\xi\X{\omega} \Longrightarrow \dot{c}=-\iota_\xi\mathrm{v}(0) \Longrightarrow c(t)=c(0) - \iota_\xi\mathrm{v}(0)t
\]
since $\omega(t)=\omega(0) + (X_\omega)t$ can be fixed at time $t=1$ to $\omega(1)=\tilde\omega$, implying that $(X_\omega)=\mathrm{v}(0)$. Similarly
\[
\X{\underline{e}^\dag}=\iota_\xi\X{\omega}\omega^\dag\Longrightarrow \dot{\underline{e}^\dag}=-\mathrm{v}(0)\omega^\dag \Longrightarrow \underline{e}^\dag(t)=\underline{e}^\dag(0) - \mathrm{v}(0)\omega^\dag t
\]
from which we can set at time $t=1$ the transformations of the fields: $\widetilde{\underline{e}}^\dag=\underline{e}^\dag(1)=e\otimes e^\dag - \mathrm{v}\otimes\omega^\dag$ and $\widetilde{c}=c(1)=c - \iota_\xi\mathrm{v}$. Notice that $\mathrm{v}\in T_\omega\mathcal{A}_{P_{\partial M}}$ is a global one-form.

The field $T_\gamma[\omega]$ is transformed as in the classical case \cite{CS4}, while $\xi$ and $e$ are projected \emph{verbatim}. Pre-composing with the obvious restriction map $\mathcal{F}_{PCH^\partial}\longrightarrow \widetilde{\mathcal{F}}_{PCH^\partial}$ yields the temporary expression
\begin{equation}\label{PCHmap1}
\underline{\pi_{PCH^\partial}}\colon\begin{cases}
\te=e\\
\txi{} =\xi\\
\tom^\dag=\omega^\dag\\
\tc=c - \iota_\xi\mathrm{v}\\
\tedl=e\otimes e^\dag - \mathrm{v}\otimes\omega^\dag\\
\tom=\omega-\mathrm{v}
\end{cases}
\end{equation}
and in the dynamical basis $\{e_\mu\}$ we have $\tom\in\mathcal{C}_{(1,2)}$. The correct ansatz for the boundary one-form in this coordinate chart reads
\begin{equation}
\underline{\alpha^\partial} = \intl_M \hat{T}_\gamma\left[-\te \te\delta \tom + \delta \tc\tom^\dag -\delta\tom\iota_{\txi{}}\tom^\dag - \iota_{\delta\txi{}}\tedl\right]
\end{equation}
as we can easily check that $\underline{\pi_{PCH^\partial}}^*\underline{\alpha^\partial}=\widetilde{\alpha}$. We can then introduce new fields redefinitions through a symplectomorphism $\phi: \mathcal{F}_{PCH^\partial}^\partial\longrightarrow \mathcal{F}_{PCH^\partial}^\partial$ as

\begin{equation}
\begin{array}{cc}
\be=-	\frac12\te\wedge\te + \iota_{\txi{}}\tom^\dag & \bom = \tom\\
\bc=\tc & \bod=\tom^\dag\\
\bxi=\txi{} & \bedl=\tedl 
\end{array}
\end{equation}

With a simple computation, using that $\omega=\bom+\mathrm{v}$, it is possible to check that the boundary one form
\begin{equation}
\alpha^\partial \coloneqq \intl_{\partial M} \mathrm{Tr}\left[\be\delta \bom + \delta \bc\bom^\dag - \iota_{\delta\bxi{}}\bedl\right]
\end{equation}
satisfies 
\begin{equation}
\widetilde{\alpha}=\pi_{PCH^\partial}^*\alpha^\partial
\end{equation}
where $\pi_{PCH^\partial}=\phi\circ\underline{\pi_{PCH^\partial}}$.

To compute the boundary action we adopt the same procedure that was shown in \cite{CS1}, following Roytenberg \cite{Royt}, namely we will compute the pre-boundary action first, by computing $\widetilde{S}\coloneqq\iota_{\widetilde{Q}}\iota_{\widetilde{E}}\widetilde{\varpi}$, where the pre-boundary graded Euler vector field $\widetilde{E}$ reads\footnote{Observe that we only consider the ghost number, and not the grading in $\Omega^\bullet(M)$ or $\bigwedge^\bullet V$.}:
\begin{equation}
\widetilde{E}=\intl_{\partial M} c\pard{}{c} + \xi\pard{}{\xi} - \omega^\dag\pard{}{\omega^\dag} - e^\dag\pard{}{e^\dag}
\end{equation}
This yields the explicit expression
\[
\widetilde{S}=\intl_{\partial M} \hat{T}_\gamma\left[\frac12 c d_\omega(ee) + \frac12 F_\omega\iota_\xi(ee) +\iota_\xi d_\omega c\omega^\dag - \frac12[c,c]\omega^\dag - \frac12\iota_\xi\iota_\xi F_\omega\omega^\dag - \frac12 \iota_{[\xi,\xi]}\underline{e}^\dag\right]
\]
and it can be checked that $\iota_{\widetilde{Q}}\widetilde{\varpi}=\delta\widetilde{S}$. Then, using the ansatz 
\[
S^\partial=\intl_{\partial M} \hat{T}_\gamma\left[ - \bc d_{\bom}\be +\iota_{\bxi}F_{\bom} \be + \frac12\left(\iota_{\bxi}\iota_{\bxi}F_{\bom} - [\bc,\bc]\right)\bom^\dag - \frac12\iota_{[\bxi,\bxi]}\bedl \right]
\]
we compute
\begin{multline}
\pi_{PCH^\partial}^*{S}^\partial=\intl_{\partial M}\hat{T}_\gamma\Bigg[\frac12c d_{\tom}ee - \frac12\iota_\xi\mathrm{v}d_{\tom}(ee) -(c-\iota_\xi\mathrm{v})d_{\tom}( \iota_\xi\omega^\dag) - \iota_\xi F_{\tom}\left(\frac12ee - \iota_\xi \omega^\dag\right) \\
+\frac12\iota_\xi\iota_\xi F_{\tom} - \frac12[c,c]\omega^\dag  +[c,\iota_\xi\mathrm{v}]\omega^\dag - \frac12[\iota_\xi\mathrm{v},\iota_\xi\mathrm{v}]\omega^\dag -\frac12\iota_{[\xi,\xi]}\underline{e}^\dag +\frac12\iota_{[\xi,\xi]}\mathrm{v}\omega^\dag\Bigg].
\end{multline}
We then use that (omitting $\hat{T}_\gamma$)
$$\intl_{\partial M}\iota_\xi\mathrm{v} d_{\tom}\iota_\xi\omega^\dag=-\intl_{\partial M}\iota_\xi d_{\tom}\iota_\xi{\mathrm{v}}\omega^\dag = -\intl_{\partial M}\frac12\iota_{[\xi,\xi]}{\mathrm{v}}\omega^\dag + \frac12\iota_\xi\iota_\xi d_{\tom}{\mathrm{v}}\omega^\dag$$
together with 
$$\intl_{\partial M}-cd_{\tom}\iota_\xi\omega^\dag +[c,\iota_\xi\mathrm{v}]\omega^\dag= \intl_{\partial M}\iota_\xi d_{\tom}c\omega^\dag +[c,\iota_\xi\mathrm{v}]\omega^\dag = \intl_{\partial M}\iota_\xi d_\omega c\omega^\dag$$
 and 
$$\frac12(\iota_\xi\iota_\xi d_{\tom}\mathrm{v} + [\iota_\xi\mathrm{v},\iota_\xi\mathrm{v}] + \iota_\xi \iota_\xi d_{\tom}\tom)\omega^\dag=\frac12\iota_\xi\iota_\xi F_\omega\omega^\dag$$
to conclude that $\pi_{PCH^\partial}^*{S}^\partial=\widetilde{S}$, concluding the proof.
\end{proof}

Observe that the Hamiltonian vector field of ${S}^\partial$ is given by the following:
\begin{equation}
\begin{array}{cc}
Q^\partial\be=-[\bc,\be] - d_{\bom} \iota_{\bxi}\be + \frac12d_{\bom}\iota_{\bxi}\iota_{\bxi}\bom^\dag & Q^\partial\bom=\iota_{\bxi}F_{\bom} - d_{\bom}\bc\\
\\
Q^\partial\bc=\frac12\left(\iota_{\bxi}\iota_{\bxi}F_{\bom} - [\bc,\bc]\right) & Q^\partial\bom^\dag = - d_{\bom}\be - [\bc,\bom^\dag] \\
\\
Q^\partial\bxi=\frac12[\bxi,\bxi] & Q^\partial\bedl_a=-(F_{\bom})_a\be - \iota_\xi (F_{\bom})_a\bom^\dag + L_\xi\bedl_a
\end{array}
\end{equation}
The general Theorems presented in \cite{CMR1} ensure us that $[Q^\partial,Q^\partial]=0$, but it is worth the while unpacking these expressions.

Since $e\wedge e = \iota_{\bxi}\bom^\dag - \be$ we can easily compute:
\begin{equation}\begin{aligned}
eQ^\partial e =& \frac12\iota_{[\bxi,\bxi]}\bom^\dag + \iota_{\bxi}Q^\partial\bom^\dag - Q^\partial\be\\
=&\iota_{\bxi}d_{\bom}\iota_{\bxi}\bom^\dag - \frac12\iota_{\bxi}\iota_{\bxi}d_{\bom}\bom^\dag - \frac12d_{\bom}\iota_{\bxi}\iota_{\bxi}\bom^\dag+\\
&\iota_{\bxi}\left(- d_{\bom}\be - [\bc,\bom^\dag]\right) +[\bc,\be] - d_{\bom} \iota_{\bxi}\be + \frac12d_{\bom}\iota_{\bxi}\iota_{\bxi}\bom^\dag\\
=&\frac12 L_{\bxi}^{\bom}(e\wedge e) - \frac12 [c,e\wedge e] = e\left( L_{\bxi}^{\bom} e - [c,e] \right)
\end{aligned}
\end{equation}
showing that the action of $Q$ on the (non-Darboux) field $e$ is essentially the action of diffeomorphism twisted by the $\mathfrak{so}(3,1)$ action, as we expect. 

Moreover, with a simple computation we get
\begin{equation}
Q^\partial(Q^\partial\be)=\frac12\left\{d_{\bom}\iota_{\bxi}L_{\bxi}^{\bom}\iota_{\bxi} \bom^\dag + \left[F_{\bom},\iota_{\bxi}\iota_{\bxi}\iota_{\bxi}\bom^\dag\right] + [\iota_{\bxi}F_{\bom},\iota_{\bxi}\iota_{\bxi}\bom^\dag]\right\}.
\end{equation}
One can easily prove the following useful identity on differential forms:
\begin{multline*}
-d_{\bom}\iota_{\bxi}L_{\bxi}^{\bom}\iota_{\bxi}  - \iota_{\bxi}\iota_{\bxi}d_{\bom}\iota_{\bxi}d_{\bom} + \iota_{\bxi}d_{\bom}\iota_{\bxi}\iota_{\bxi}d_{\bom} +[\iota_{\bxi}\iota_{\bxi}F_{\bom},\iota_{\bxi}] + [\iota_{\bxi}F_{\bom},\iota_{\bxi}\iota_{\bxi}] = \\
=-d\iota_{\bxi}L_{{\bxi}}\iota_{\bxi} - \iota_{\bxi}\iota_{\bxi}d\iota_{\bxi}d + 2\iota_{\bxi} d\iota_{\bxi}\iota_{\bxi} d = \frac12\iota_{\left[[\bxi,\bxi],\bxi\right]}=0,
\end{multline*}
the application of which to the top-form $\bom^\dag$ yields:
\[
-d_{\bom}\iota_{\bxi}L_{\bxi}^{\bom}\iota_{\bxi}\bom^\dag +[\iota_{\bxi}\iota_{\bxi}F_{\bom},\iota_{\bxi}\bom^\dag] + [\iota_{\bxi}F_{\bom},\iota_{\bxi}\iota_{\bxi}\bom^\dag] \equiv 0
\]
Considering now that $F_{\bom}\wedge\iota_{\bxi}\bom^\dag$ is a 4-form and therefore vanishes on the boundary, using all of the above we conclude
\[
0=\iota_{\bxi}\iota_{\bxi}[F_{\bom},\iota_{\bxi}\bom^\dag] = [\iota_{\bxi}\iota_{\bxi}F_{\bom},\iota_{\bxi}\bom^\dag] + 2[\iota_{\bxi}F_{\bom},\iota_{\bxi}\iota_{\bxi}\bom^\dag] + [F_{\bom},\iota_{\bxi}\iota_{\bxi}\iota_{\bxi}\bom^\dag] \equiv Q^\partial(Q^\partial\be).
\]

\begin{remark}
The BFV-manifold obtained in Theorem \ref{Theo:PCHpres2} is the resolution of the coisotropic submanifold $C_{BP}$ of Theorem \ref{constraintstheorem}, defined by the equations (compare with \cite{CS1,CS4}):

\begin{align}
C_e\colon Q^\partial \bom^\dag\big|_{\mathrm{gh}=0} \equiv d_{\bom} e\wedge e = 0\\
C_{{\bom}}\colon Q^\partial \bedl\big|_{\mathrm{gh}=0} \equiv F_{\bom}\wedge e\wedge e = 0 
\end{align}

As a matter of fact, observe that the degree 0 part of the action of $Q^\partial$ on $\be$ coincides with the Hamiltonian vector field of the constraint $\mathbf{J}_{\iota_\xi\bem}$, namely (cf. with Eq. \eqref{hamiltonBP})
$$\bem\wedge(\mathsf{J}_{\iota_\xi\bem})_\bem\equiv (\mathsf{J}_{\iota_\xi\bem})_{\bem\wedge\bem} = \frac12d_{\bom}\iota_\xi(\bem\wedge\bem) = Q^\partial\be\big|_{gh=0}$$
and similarly for the action on $\bom$. However, this is inequivalent to the Einstein--Hilbert phase space for, as we mentioned, the Hamiltonian constraint is not taken into account.
\end{remark}


\begin{thebibliography}{99}
\addcontentsline{toc}{chapter}{Bibliography}
\bibitem[Ati]{Ati} M. Atiyah, \emph{Topological quantum field theories}, Inst. Hautes Etudes Sci. Publ. Math. {\bf 68}, 175-186 (1988).
\bibitem[BB]{BB} L. Baulieu and M. Bellon, {\it p-Forms and Supergravity: gauge symmetries in curved space}, Nucl. Phys. {\bf B 266}, 75-124 (1986).

\bibitem[BF83]{BFV1} I. A. Batalin and E. S. Fradkin, \emph{A generalized canonical formalism and quantization of reducible gauge theories}, Phys. Lett. {\bf B 122}(2), 157-164 (1983).
\bibitem[BV77]{BFV2} I. A. Batalin and G. A. Vilkovisky, \emph{Relativistic S-matrix of dynamical systems with boson and fermion costraints}, Phys. Lett. {\bf B 69}(3), 309-312 (1977).
\bibitem[BV81]{BV81} I. A. Batalin and G. A. Vilkovisky. {\it Gauge algebra and quantization}, Phys. Lett. {\bf B 102}(1), 27-31 (1981).


\bibitem[Car]{Cartan} E. Cartan. {\it Sur une généralisation de la notion de courbure de Riemann et les espaces à torsion}, C. R. Acad. Sci. {\bf 174}, 593?595 (1922).\\
E. Cartan, Comptes rendus hebdomadaires des séances de l'Académie des sciences,  174, 437-439, 593-595, 734-737, 857-860, 1104-1107 (January 1922)

\bibitem[CaSc]{CanSch} G. Canepa and M. Schiavina, {\it Fully extended BV-BFV description of General Relativity in three dimensions}, arXiv:1905.09333 [math-ph].

\bibitem[CMR14]{CMR1} A. S. Cattaneo, P. Mn\"ev and N. Reshetikhin, {\it Classical BV theories on manifolds with boundary}, Comm. Math. Phys. {\bf 332}(2): 535-603 (2014).
\bibitem[CMR15]{CMR2} A. S. Cattaneo, P. Mn\"ev and N. Reshetikhin, {\it Semiclassical quantization of Lagrangian field theories}, Mathematical Aspects of Quantum Field Theories, in Mathematical Physics Studies 2015, pp 275-324.
\bibitem[CMRQ]{CMR3} A. S. Cattaneo, P. Mn\"ev and N. Reshetikhin, {\it Perturbative quantum gauge theories on manifolds with boundary}, Commun. Math. Phys. {\bf 357}, 631?730 (2018).
\bibitem[CMR11]{CMRCorfu} A. S. Cattaneo. P. Mnev and N. Reshetikin, Classical and Quantum Lagrangian Field Theories with Boundary, Proceedings of the Corfu Summer Institute 2011 School and Workshops on Elementary Particle Physics and Gravity, Corfu, Greece, 2011.

\bibitem[CS15]{CS1} A. S. Cattaneo and M. Schiavina, {\it BV-BFV approach to General Relativity: Einstein--Hibert action}, J. Math. Phys. {\bf 57}(2) (2015).

\bibitem[CS16]{CS3} A. S. Cattaneo and M. Schiavina, {\it On time}, Lett. Math. Phys. {\bf 107}(2), 375-408 (2017).

\bibitem[CS17]{CS4} A. S. Cattaneo and M. Schiavina, {\it The reduced phase space of Palatini--Cartan--Holst theory}, Annales Henri Poicar\'e {\bf 20}(2), 445-480 (2019).

\bibitem[CSS]{CSS} A. S. Cattaneo, M. Schiavina and I. Selliah, {\it BV equivalence between triadic gravity and BF theory in three dimensions}, Letters in Mathematical Physics {\bf 108}(8), 873?1884 (2018).



\bibitem[Dir]{Dirac} P.A.M. Dirac, {\it Generalized Hamiltonian dynamics}, Canad. J. Math. {\bf 2}, 129-148 (1950).
\bibitem[FK]{FK} G. Felder and D. Kazhdan. {\it The classical master equation}, in Contemporary Mathematics, Vol {\bf 610} (2014).

\bibitem[HT]{Henn} M. Henneaux and C. Teitelboim, {\it Quantization of Gauge Systems}, Princeton university press (1994). 

\bibitem[Hol]{Holst} S. Holst, {\it Barbero's Hamiltonian derived from a generalized Hilbert-Palatini action}, Phys. Rev. {\bf D 53}, 5966 (1996).

\bibitem[MSS]{Morschsor} O. Moritsch, M. Schweda, S. P. Sorella, {\it  Algebraic structure of Gravity with Torsion}, Class. Quant. Grav. {\bf 11} 1225 (1994).


\bibitem[KT]{KT} Kijowski, J., Tulczyjew, W. M., \emph{A Symplectic Framework for Field Theories}, Lecture notes in Physics {\bf 107}, Springer-Verlag Berlin Heidelberg (1979).
\bibitem[Mn]{PavelTh} P. Mnev, {\it Discrete BF theory}, arXiv:0809.1160 (2008).




\bibitem[Pal]{Palatini} A. Palatini, {\it Deduzione invariantiva delle equazioni gravitazionali dal principio di Hamilton}, Rend. Circ. Mat. Palermo {\bf 43}, 203 (1919).

\bibitem[RP]{Per} D. J. Rezende and A. Perez, {\it Four-dimensional Lorentzian Holst action with topological terms}, Phys. Rev. {\bf D 79}, 064026 (2009).
\bibitem[Piguet]{Piguet} O. Piguet, {\it Ghost Equations and Diffeomorphism Invariant Theories}, Class. Quant. Grav. {\bf 17}, 3799-3806 (2000).

\bibitem[RT]{Rovth} C. Rovelli and T. Thiemann, {\it Immirzi parameter in quantum general relativity}, Phys. Rev. {\bf D 57}, 1009 (1998).

\bibitem[Roy]{Royt} D. Roytenberg, {\it AKSZ-BV Formalism and Courant Algebroid-induced Topological Field Theories}, Lett. Math. Phys. {\bf 79}: 143-159 (2007).

\bibitem[Scha09]{Schaetz09} F. Schaetz, {\it BFV-complex and higher homotopy structures}, Comm. Math. Phys. {\bf 286}(2), 399-443 (2009).


\bibitem[Scha10]{Schaetz10} F. Schaetz, {\it Invariance of the BFV complex}, Pac. J. Math. {\bf 248} (2), (2010).

\bibitem[Schi]{THESIS} M. Schiavina, {\it BV-BFV  approach to General Relativity}, PhD Thesis, University of Z\"urich (2016).

\bibitem[Seg]{Seg} G. Segal, \emph{The definition of conformal field theory}, in: Differential geometrical methods in theoretical physics, Springer Netherlands, 165-171 (1988).

\bibitem[Sta96]{Stash} J.Stasheff, {\it Deformation Theory and the Batalin-Vilkovisky Master Equation}, Deformation theory and symplectic geometry, proceedings, Meeting, Ascona, Switzerland, June 16-22, 1996, arXiv:q-alg/9702012.

\bibitem[Sta97]{stash97} J. Stasheff, {\it Homological reduction of constrained Poisson algebras}, J. Diff. Geom. {\bf 45}, 221-240 (1997).





\end{thebibliography}
\end{document}